\documentclass[11pt]{article}
\usepackage{amsthm,amsmath,amsfonts,amssymb} 
\usepackage{mathptmx}
\usepackage{microtype}
\usepackage{fullpage}
\usepackage{doi}
\usepackage{natbib}
\usepackage{tikz}
\usepackage{caption}
\usepackage{subcaption}
\usepackage{eucal}

\setcitestyle{numbers,square}

\newtheorem{theorem}{Theorem}
\newtheorem{lemma}[theorem]{Lemma}

\newtheorem{corollary}[theorem]{Corollary}
\newtheorem{definition}{Definition}

\DeclareMathOperator*{\Exp}{\mathbb{E}}
\newcommand{\N}{\mathbb{N}}
\newcommand{\Z}{\mathbb{Z}}
\newcommand{\encode}[1]{\mathsf{C}(#1)}
\newcommand{\sencode}[1]{\mathsf{R}(#1)}
\newcommand{\knuth}[1]{\mathsf{K}(#1)}
\newcommand{\oddpairs}{\mathbin{\lozenge_1}}
\newcommand{\evenpairs}{\mathbin{\lozenge_0}}
\newcommand{\sevenpairs}{\mathbin{\blacklozenge_0}}
\newcommand{\soddpairs}{\mathbin{\blacklozenge_1}}
\newcommand{\wt}{\operatorname{wt}}
\newcommand{\ilog}{\operatorname{{\log}^\sharp}}
\newcommand{\shift}{\textsf{S}}
\newcommand{\unique}[1]{\textsf{U}(#1)}
\newcommand{\maximal}[1]{\textsf{M}(#1)}
\newcommand{\setsystem}[1]{\mathcal{#1}}
\newcommand{\Hide}[1]{}

\title{Deterministic Blind   Rendezvous\\  in
  Cognitive Radio Networks}

\author{Sixia Chen\qquad Alexander Russell\\Department of Computer
  Science and Engineering\\ University of Connecticut \and Abhishek
  Samanta \qquad Ravi Sundaram\\ College of Computer Science\\
  Northeastern University}

\begin{document}
\maketitle
\thispagestyle{empty}

\begin{abstract}
\emph{Blind rendezvous}  is a  fundamental problem in  cognitive radio
networks. The  problem involves a  collection of agents  (radios) that
wish  to  discover each  other  (i.e.,  \emph{rendezvous}) in  the
\emph{blind} setting where there  is no shared infrastructure and they
initially  have no  knowledge of  each  other.  Time  is divided  into
discrete slots and spectrum is  divided into discrete channels, $[n] =
\{1,2,\ldots, n\}$. Each agent may access  (or \emph{hop on})
 a single channel in
a single time slot and two agents \emph{rendezvous}  when they
hop on
the  same channel  in  the same  time  slot.  The  goal  is to  design
deterministic channel hopping  schedules  for  each agent  so  as  to
guarantee rendezvous between any pair of agents with access
to overlapping sets of channels.

The problem  has three complicating considerations:  first, the agents
are  \emph{asymmetric}, i.e.,  each agent  $A_i$ only  has access  to a
particular  subset $S_i  \subset [n]$  of the  channels  and different
agents may have access to  different subsets of channels (clearly, two
agents can rendezvous only  if their channel subsets overlap); second,
the agents are \emph{asynchronous}, i.e., they do not possess a common
sense of absolute time, so different agents may commence their channel
schedules  at different times  (they do  have a  common sense  of slot
duration);  lastly,  agents are  \emph{anonymous}  i.e.,  they do  not
possess an identity, and hence the schedule for $A_i$ must depend only
on $S_i$.

Whether guaranteed blind rendezvous in the asynchronous model was even
achievable  was  an  open  problem.   In a  recent  breakthrough,  two
independent sets of authors, Shin et al.~[Communications Letters 2010]
and  Lin   et  al.~[INFOCOM   2011],  gave  the   first  constructions
guaranteeing  asynchronous blind rendezvous  in $O(n^2)$  and $O(n^3)$
time,   respectively.   We  present   a  substantially   improved  and
conceptually  simpler construction guaranteeing  that any  two agents,
$A_i$, $A_j$, will rendezvous  in $O(|S_i| |S_j|\log\log n)$ time. Our
results are  the first that achieve nontrivial  dependence on $|S_i|$,
the  size  of the  set  of available  channels.  This  allows us,  for
example, to  save roughly  a quadratic factor  over the  best previous
results  in the  important  case when  channel  subsets have  constant
size. We  also achieve  the best possible  bound of  $O(1)$ rendezvous
time for  the symmetric situation;  previous works could do  no better
than  $O(n)$.  Using  the probabilistic  method and  Ramsey  theory we
provide evidence in support of  our suspicion that our construction is
asymptotically  optimal  (upto   constants)  for  small  size  channel
subsets:  we show  both an  $\Omega(|S_i| |S_j|)$  lower bound  and an
$\Omega(\log\log n)$ lower bound when $|S_i|, |S_j| \leq n/2$.
\end{abstract}
\newpage
\pagenumbering{arabic}

\section{Introduction}

 \subsection{Motivation}
 Given   the   ever-increasing  demand   for   all  things   wireless,
 \emph{spectrum}   has  become   a  scarce   resource.   Historically,
 regulators  around the  world  have employed  a  command and  control
 philosophy  towards managing  spectrum \cite{wikispectrummanagement}:
 Some  channels  were statically  licensed  to  particular users  (for
 certain periods  and in certain  geographies) while others  were kept
 aside  for  community  use.   \emph{Cognitive  radio  networks}  have
 emerged  as  a  modern,   dynamic  approach  to  spectrum  allocation
 \cite{Zhao07asurvey,Akyildiz06dynamic}.        Exploiting      recent
 technological  developments,  cognitive  agents (radios)  dynamically
 sense incumbent  users and opportunistically hop  to unused channels.
 While  they   can  offer  improved  utilization,   they  introduce  a
 fundamental \emph{rendezvous} problem: the problem of discovering the
 existence of peers in a multichannel setting.


\subsection{Model and  Results}
We work in  the \emph{blind} model where a  collection of agents $A_i$
wish to discover  each other with no dedicated  common control channel
or other  shared infrastructure.  Time is divided  into discrete slots
and spectrum is divided into discrete channels, $[n] = \{1,2,\ldots, n\}$.
Each agent  may access (or hop on)  a single channel in  a single time
slot  and two  agents  \emph{rendezvous}  when they  hop  on the  same
channel  in  the  same  time  slot.   The challenge  is  to  design  a
channel-hopping  schedule for each  agent so  that they  discover each
other. As stated thus far,  the problem has the trivial solution where
all agents can  hop on a specific channel, say channel  1, in the very
first time slot.  However,  reality is complicated by three additional
requirements:       \emph{asymmetry},       \emph{asynchrony}      and
\emph{anonymity}.

\noindent
{\bf Asymmetry} Different agents  may have access to different subsets
of channels as  a result of local interference  or variations in radio
capabilities.  Let  $S_i \subseteq [n]$  be the subset of  channels to
which agent  $A_i$ has  access.  Thus the  challenge is to  create for
each agent $A_i$ a  channel-hopping \emph{schedule} $\sigma_i: \{0, 1,
\ldots\}  \rightarrow   S_i$  which  guarantees   that  $\exists\,  t,
\sigma_i(t)  = \sigma_j(t)$  for any  two agents  $A_i,\,  A_j, \mbox{
  s.t. } S_i  \cap S_j \neq \emptyset$. (In  the symmetric setting all
agents have access to the identical subset of channels.)

\noindent
{\bf Asynchrony}  Different agents  may not share  a common  notion of
time.   They may commence  at different  ``wake-up'' times  inducing a
relative shift  in their progress  through their schedules.  Note that
agents do possess a common  understanding of slot duration.  The goal,
therefore, is  to ensure  rendezvous between a  pair of agents  in the
shortest  possible  time  once  they  have  both  woken  up.  (In  the
synchronous  setting all  agents  share a  common  notion of  absolute
time.)

\noindent
{\bf Anonymity} In our setting an  agent's schedule must depend only on
the subset  of channels available to,  and not on  a distinct identity
of, the  agent i.e., $\sigma_i$ must  depend only on  $S_i$. Note that
$S_j$ is  unknown to $A_i$ for  $i \neq j$  and it is allowed  for two
different agents  to have the  same set of accessible  channels, i.e.,
$S_i = S_j$ for $i \neq j$.

Now,  the problem  has the  naive randomized  solution, in  which each
agent,   at  each  time   step,  selects   a  channel   uniformly  and
independently at random  from its subset.  It is not  hard to see that
this provides  a high-probability  guarantee of rendezvous  for agents
$A_i,\,   A_j$  in  time   $O(|S_i|  |S_j|   \log  n)$.    However,  the
deterministic  setting is  the  gold-standard in  the cognitive  radio
networking  community:  it makes  the  weakest  assumptions about  the
devices, which  need not have  an available source of  randomness, and
provides absolute guarantees on rendezvous time.

Here we briefly summarize of our main results:\\
\noindent
{\bf Algorithms}
  \begin{enumerate}

  \item\label{item:two} We give an  $O(\log\log n)$ time algorithm for
    rendezvous for the special case of agents with $|S_i| = 2$.

  \item We then  show how to apply this  algorithm to yield algorithms
    for  arbitrary subsets  of $[n]$  that guarantees  rendezvous time
    $O(|S_i|  |S_j| \log  \log n)$  for all  pairs of  sets  $S_i$ and
    $S_j$.

  \item We show that a minor adaptation of this algorithm can furthermore guarantee $O(1)$ time rendezvous for the symmetric case.
    
  \item Finally, we explore the ``one bit beacon'' case, where the
    agents have the luxury of a single common random bit during each
    time slot. In this model, we show that $O(|S_i| + |S_j| + \log n)$
    time is sufficient, with high probability, to rendezvous.
\end{enumerate}
\smallskip

\noindent{\bf Lower Bounds}
  \begin{enumerate}
  \item We prove an $\Omega(\log\log n)$ lower bound on the rendezvous
    time, even for \emph{synchronous} agents with the promise that the
    channel  sets  $S_i$ have  constant  size.  This  shows that  some
    dependence on  $n$, the  size of the  channel universe,  is always
    necessary. 
    In particular, this shows that
    the algorithm of~\ref{item:two} above is tight up to a constant.
  \item For channel  subsets of size $k$ we prove  a $k^2$ lower bound
    on even the synchronous rendezvous time, under the promise that $k
    = O(\log  n/\log\log n)$.  For larger values  of $k$, we  obtain a
    weaker family of results.
  \item  In the asynchronous  time model,  we prove  that $|S_i||S_j|$
    steps are necessary to rendezvous,  so long as $|S_i| + |S_j| \leq
    n$.
  \end{enumerate}

  We  also consider  a one-round  version of  the problem;  instead of
  minimizing  the  number  of   rounds  we  consider  the  problem  of
  maximizing the number of pairs of agents that can achieve rendezvous
  in a  single round. In  particular for the ``graphical''  case where
  channel sets are  of size 2 we show how a  variant of the celebrated
  Goemans-Williamson   semi-definite  program   \cite{GoemansW95}  for
  MAX-CUT  can be  employed to  obtain a  0.439 approximation  for the
  one-round  maximization version.  This  result is  presented in  the
  Appendix.

\subsection{Related work} 
Rendezvous problems  have a  long history in  mathematics and computer science---an early
example    is    Rado's    famous    ``Lion    and    Man''    problem
\cite{littlewoodmiscellany}.   Over  time a  variety  of problems  and
solutions  have  evolved in  both  adversarial \cite{rufusisaacs}  and
cooperative settings \cite{alperngal}. Rendezvous in networks has been
extensively    studied    in    the   computer    science    community
\cite{PelcSurvey}. Though  the study of rendezvous  in cognitive radio
networks  is relatively  recent there  already exists  a comprehensive
survey  \cite{surveycrn}  that contains  a  detailed  taxonomy of  the
different models including the specific one relevant to this work. The
problem of guaranteed blind rendezvous in the asymmetric, asynchronous
and  anonymous  case  was  first considered  in  \cite{seq-based}  and
subsequently in \cite{2n+1-seq, ring-walk}.   The use of prime numbers
and modular algorithms was  initiated in \cite{rend-crn}. However, the
general case of the  problem withstood attack until \cite{shinyangkim,
  LLCL:Jump}. The  current state of the art  is \cite{secon2013} which
achieves an $O(n^2)$ algorithm for  the asymmetric case and $O(n)$ for
the symmetric  case. A crucial difference  between these constructions
and  ours is that  we explicitly  exploit the  fact that  the schedule
$\sigma_i$  can  depend  arbitrarily  on $S_i$,  whereas  the  earlier
constructions  \cite{shinyangkim,  LLCL:Jump,  secon2013}  derive  the
schedule  for a channel  subset by  (essentially) projecting  onto the
desired subset from a single uniformly generated schedule for the full
set of  channels.  Our  work is notable  for providing  a conceptually
clean  and  significantly  more  efficient  $O(|S_i||S_j|\log\log  n)$
algorithm  for the general  asymmetric setting.   Real-world cognitive
networks  \cite{YucekA09}  create  a  pooled  hyperspace  occupied  by
signals with  dimensions of frequency, time, space,  angle of arrival,
etc., created  by advances in antenna design,  and comprising spectrum
ranging from radio frequencies and TV-band white spaces to lasers.  In
these networks  the total number of  channels, $n$ is  large while the
channel subsets accessible  to any given device are  small.  A similar
situation prevails in military situations where different members of a
(dynamic)  coalition  operate in  a  small  portion  of the  available
spectrum that  guarantees overlap with allies. In  such situations our
scheme  achieves  a  near-quadratic  factor  gain  over  the  previous
results.   And for  the  symmetric setting  our construction  achieves
$O(1)$   rendezvous   time   which   clearly   cannot   be   bettered.
Table~\ref{table:results} presents  a summary  of our upper  bounds in
the context of prior work.

\begin{table}[ht]
\caption{Upper bounds for deteministic rendezvous}
\centering
\begin{tabular}{|l|c|c|}
\hline
Paper & Asymmetric & Symmetric \\
\hline \hline
Shin-Yang-Kim~\cite{shinyangkim} & $O(n^2)$ & $O(n^2)$ \\
Lin-Liu-Chu-Leun~\cite{LLCL:Jump} & $O(n^3)$ & $O(n)$ \\
Gu-Hua-Wang-Lau~\cite{secon2013} & $O(n^2)$ & $O(n)$ \\
Our results & $O(|S_i||S_j|\log\log  n)$ & $O(1)$ \\
\hline
\end{tabular}
\label{table:results}
\end{table}

We are also
the  first to  provide nontrivial  lower bounds  employing  tools from
Ramsey theory and the  probabilistic method.  \cite{GasieniecPP01} is a
closely related work; its globally synchronous and locally synchronous
models   correspond  to  our   asynchronous  and   synchronous  models
respectively.  However, while \cite{GasieniecPP01}  explicitly requires
that  exactly one  node  transmits on  a  single fixed  channel for  a
successful broadcast,  we implicitly assume  that once a  set achieves
rendezvous  (on any  one of  several  channels) then  they employ  the
standard  ``chirp  and  listen'' technique  \cite{ZhangLuoGuo2013}  to
ensure mutual identification of the set members.

\section{Definitions and notation}

Let  $\setsystem{S}$   be  a  collection  of  subsets   of  $[n]$.  An
\emph{$\setsystem{S}$-schedule} is a family of schedules $\sigma_S: \N
\rightarrow S$, one for each  $S \in \setsystem{S}$. In fact, we focus
solely on two special cases:
  \begin{itemize}
  \item  An  \emph{$n$-schedule}  is  a $2^{[n]}$-schedule,  one  that
    supplies a schedule for every subset of $[n]$.
  \item  An  \emph{$(n,k)$-schedule}  is  a  $\setsystem{S}$-schedule,
    where  $\setsystem{S}$ consists of  all subsets  of $[n]$  of size
    $k$.
  \end{itemize}
  We will  typically reserve the notation $\Sigma  = (\sigma_A)_{A \in
    \setsystem{S}}$  to denote an  $\setsystem{S}$-schedule; departing
  from the notation used  in the introduction, the schedule associated
  with the set $A$ is simply denoted $\sigma_A$.
  
  Let $\sigma_A: \N \rightarrow A$ and $\sigma_B: \N \rightarrow B$ be
  two schedules for  overlapping subsets $A$ and $B$  of $[n]$. We say
  that $\sigma_A$ and $\sigma_B$  \emph{rendezvous synchronously in time $T$}
  if there is  a time $t \leq T$ so  that $\sigma_A(t) = \sigma_B(t)$;
  this  corresponds  to  rendezvous   in  the  synchronous  model
  discussed  in the  introduction. Recall  that the  asynchronous
  model introduces  arbitrary ``wake-up''  times $t_A$ and  $t_B$ into
  each  of the  two schedules,  after  which they  proceed with  their
  schedules. Of  course, in this case they  cannot possibly rendezvous
  before time $\max(t_A, t_B)$,  when they are finally both ``awake.''
  Thus, we  say that these two schedules  \emph{rendezvous asynchronously in 
    time  $T$}  if,  for all  $t_A,  t_B  \geq  0$,  there is  a  time
  $\max(t_A,t_B)   \leq   t   \leq   \max(t_A,t_B)  +   T$   so   that
  $\sigma_A(t-t_A) = \sigma_B(t-t_B)$.
  
  For a fixed $(n,k)$-schedule $\Sigma$, we define $R_s(\Sigma)$ to be
  the minimum  $T$ for which  $\sigma_A$ and $\sigma_B$  synchronously rendezvous in
  time  $T$ for all $A,  B \in \setsystem{S}$.  We   likewise   define   
  $R_a(\Sigma)$  for   asynchronous rendezvous. Finally, we define:
  \[
  R_s(n,k) \triangleq \min_{\Sigma} R_s(\Sigma)\quad\text{and}\quad
  R_a(n,k) \triangleq \min_{\Sigma} R_a(\Sigma)\,,
  \]
  where these  are minimized  over all $(n,k)$-schedules  $\Sigma$. Of
  course,  $R_s(n,k)   \leq  R_a(n,k)$,  and   the  simple  randomized
  algorithm described in the introduction suggests that perhaps
  \[
  R_a(n,k) \approx k^2 \,.
  \]
  Finally, we  remark that even a precise  understanding of $R_a(n,k)$
  does   not   necessarily   yield  $n$-schedules   that   guarantee
  satisfactory  bounds on pairwise  rendezvous because  it is  not, in
  general,  clear   how  to  stitch   together  $(n,k)$-schedules  for
  different values of  $k$ to provide guarantees for  pairs of sets of
  different sizes.
  
\paragraph{Notation} We use $[n] = \{ 1, \ldots, n\}$ and invent the shorthand notation
$\ilog n \triangleq \lceil \log_2 n \rceil$. 
Whenever a variable, $x$, represents a natural number, we use $x_2$ to denote the 
canonical base-two encoding of $x$, zero-padded on the left out to length $\ilog m$, 
where $m$ is the maximum value that $x$ might take.




\section{Schedules for efficient rendezvous}\label{sec:constructions}

\paragraph{Sets of size two}
We begin with a construction of a family of schedules for channel sets
of size 2 that achieves rendezvous in time $O(\log \log n)$; these
will be used as a subroutine for the general construction.  We shall
see in Section~\ref{sec:lower-bounds} that these schedules are within
a constant of optimal.  Thus, the goal of this section is to prove the
following theorem.

\begin{theorem}\label{thm:upper-2}
  For all $n > 0$, $R_a(n,2)  = O(\log \log n)$. Specifically, for any
  $n > 0$, there  is an $(n, 2)$-schedule so that for  any two sets $A$
  and  $B$   of  size   two,  $\sigma_A$  and   $\sigma_B$  rendezvous
  asynchronously in time no more than $O(\log \log n)$.
\end{theorem}

The size 2 construction is based  on the remarkable fact that there is
an edge coloring of the linear  poset, using only $\ilog n$ colors, for
which no  path of length two is  monochromatic. Specifically, consider
the directed  graph $L_n =  (V_n, E_n)$, with  vertex set
  $V_n = [n]$ and directed edges $E_n = \{ (a, b) \mid a < b\}$.
A  \emph{2-Ramsey edge  coloring} of  $L_n$  is a  mapping $\chi:  E_n
\rightarrow P$  with the property that $\chi(a,b)  \neq \chi(b,c)$ for
any pair of  directed edges $(a,b)$ and $(b, c)$  that form a directed
path of length 2.

\begin{lemma}\label{lem:coloring}
The graph  $L_n$ has a 2-Ramsey  edge coloring with a  palette of size
$\ilog n$.
\end{lemma}

\begin{proof}
With hindsight, associate with each vertex $k \in V_n$ the set
\[
X_k = \{ i \mid \text{the $i$th bit of $k_2$ is a 1}\} \subset \{ 1,
\ldots, \ilog n\}\,.
\]
Observe that if
$a < b$, there is an element  in $X_b \setminus X_a$. In this case, we
may safely color  the edge $(a,b)$ with any  element of $X_b \setminus
X_a$,  as it  follows immediately  that any  pair of  edges  forming a
directed path must have distinct  colors. The scheme uses no more than
$\ilog n$ colors.
\end{proof}

\begin{proof}[Proof of Theorem~\ref{thm:upper-2}]
 We begin with  a construction for the simpler  synchronous model, and
 then show how to reduce the asynchronous model to this case.

\smallskip
\noindent
\textsl{The synchronous model.}
In the synchronous model, we will simplify the presentation by
discussing finite length schedules with the understanding that
rendezvous is guaranteed by the time the schedule has been exhausted.
Consider now a subset of two channels $A = \{a_0, a_1\}$, where $a_0 <
a_1$. We will treat such size-two subsets as directed edges of the
linear poset (directed from the smaller element to the larger
element). In this size-two case,  we may express  a schedule as a binary
string $s_0 s_1 s_2 \ldots  \in \{0,1\}^*$ with the convention that at
time $t$, the  schedule calls for $a_{s_t}$: thus, when  $s_t = 0$ the
schedule calls  for the smaller of  the two channels; when  $s_t = 1$,
the schedule calls for the larger of the two channels.

Consider now a pair of overlapping subsets $A = \{ a_0, a_1\}$ and $B
= \{b_0, b_1\}$ with $a_0 < a_1$ and $b_0 < b_1$. When these two edges
form a directed path (so that their common element is the larger of
one set and the smaller of the other), a sufficient condition for two
schedules $r_0r_1 \ldots r_{\ell-1}$ and $s_0 s_1\ldots s_{\ell-1}$ to
rendezvous is that each of the two tuples $\{ (0,1), (1,0)\}$ can be
realized as $(r_t, s_t)$ for some $t$, which is to say that
\begin{equation}\label{eq:rendezvous-condition-odd}
\{ (0,1), (1,0) \} \subset \{  (r_t, s_t)  \mid 0  \leq t  < \ell  \}\,.
\end{equation}
We reserve the notation $r \oddpairs s$ to denote the statement that the
strings $r$ and $s$ satisfy
condition~\eqref{eq:rendezvous-condition-odd}. Likewise, when $\{a_0,
a_1\}$ and $\{ b_0, b_1\}$ do not form a path of length two (that is,
share a common largest or smallest element), a sufficient condition
for rendezvous is that
\begin{equation}\label{eq:rendezvous-condition-even}
\{ (0,0), (1,1) \} \subset \{  (r_t, s_t)  \mid 0  \leq t  < \ell  \}\,.
\end{equation}
We reserve the notation $r \evenpairs s$ to denote the statement that
$r$ and $s$ satisfy~\eqref{eq:rendezvous-condition-even}.

In the remainder of the proof we identify a map $x \mapsto
\encode{x}$ with the property that
\begin{align}\label{eq:encoding-property-even}
  x = y \quad &\Rightarrow \quad \encode{x} \evenpairs \encode{y}\,,\\
 \label{eq:encoding-property-odd} x \neq y \quad &\Rightarrow \quad \encode{x} \oddpairs \encode{y}\,.
\end{align}
With such a map in hand, we adopt the schedule $\encode{\chi(\alpha,
  \beta)_2}$ for the set $\{\alpha, \beta\}$, where $\chi$ is the edge
coloring of Lemma~\ref{lem:coloring}. Observe that if $A = \{ a_0,
a_1\}$ and $B = \{ b_0, b_1\}$ form a path of length two,
$\chi(a_0,a_1) \neq \chi(b_0, b_1)$ and this schedule guarantees
rendezvous by dint of
property~\eqref{eq:encoding-property-odd}. Otherwise, these schedules
guarantee rendezvous by dint of
property~\eqref{eq:encoding-property-even}.



We return to the problem of constructing the map $\encode{\cdot}$.
By adopting the convention that all schedules start with the prefix
$01$, we can immediately guarantee property~\eqref{eq:encoding-property-even}: $(0,0)$ and $(1,1)$ appear in
$\{ (r_t, s_t) \mid 0 \leq t < \ell \}$.
It  is  easy to  check  that  the map  $x  \mapsto  01  \circ x  \circ
\overline{x}$, where $\circ$  denotes concatenation and $\overline{x}$
the  coordinatewise negation  of $x$,  has the  desired  properties.

A leaner mapping can be obtained by the rule
\[
\encode{x} \triangleq 01 \circ x \circ \overline{\wt(x)_2}\,,
\]
where $\wt(x)$ denotes the weight (number of 1s) of the string $x$.  
To see that this encoding has property~\eqref{eq:encoding-property-odd}, observe that when $\wt(x) =
\wt(y)$, both $(0,1)$ and $(1,0)$ must appear in the set $\{ (x_i,
y_i) \mid 1 \leq i \leq |x|\}$ (where $x_i$ is the $i^{th}$ bit of $x$) as $x \neq y$ and they have common
weight. When $\wt(x) < \wt(y)$, it follows immediately that $(0,1) \in
\{ (x_i, y_i) \mid 1 \leq i \leq |x|\}$; as for the tuple $(1,0)$, this
must be realized by one of the coordinates of $\overline{\wt(x)_2}$
and $\overline{\wt(y)_2}$ as the canonical encoding of integers in
binary ensures that when $n < m$, there is a coordinate in which $n_2$
contains a $0$ and $m_2$ contains a $1$. The case when $\wt(x) >
\wt(y)$ is handled similarly.

Finally, we remark that when $x$ has length $\ell$, $\encode{x}$ has
length $\ell + \ilog \ell + 2$.
As $L_n$ can be edge colored with a palette of size $\ilog n$, this
yields a family of schedules for sets of size $2$ that guarantees
rendezvous in time no more than
$\ilog \ilog n + \ilog \ilog \ilog n + 2$.

\smallskip
\noindent
\textsl{The asynchronous model.}
We  return  now  to  the  asynchronous  model  described  in  the
introduction, in which  the two agents' schedules are  subjected to an
unknown  shift due  to potentially  distinct start-up  times.  In this
model, we  are obligated  to define schedules  for all  nonnegtive times
(that  is, our  schedules  have  the form  $\sigma:  \N \rightarrow  S
\subset  [n]$);   one  straightforward  method   for  describing  such
schedules is  to adopt \emph{cyclic schedules},  which cyclicly repeat
the same finite sequence of  channels. In particular, if $\sigma: \{
0, \ldots, \ell-1\} \rightarrow  S \subset [n]$, we let $\sigma^{\circ}:
\N  \rightarrow  S$ denote  the  schedule  $\sigma^{\circ}: t  \mapsto
\sigma(t \bmod \ell)$.

Continuing in the spirit of the previous discussion, we observe that
if $r = r_0 \ldots r_{\ell-1}$ and $s = s_0 \ldots s_{\ell-1}$ are two
schedules for a pair of sets $A = \{ a_0, a_1\}$ and $B = \{ b_0,
b_1\}$ forming a path, the cyclic schedules they induce will guarantee rendezvous (in
time $\ell$) if, for all $i$ and $j$,
\begin{equation}\label{eq:two-common-criterion-odd}
\shift^i r \oddpairs \shift^j s\,,
\end{equation}
where $\shift^i x$ denotes the result of cyclicly shifting $x$ forward
$i$ symbols.  To save ink, we define $r \soddpairs s$ to denote the
condition~\eqref{eq:two-common-criterion-odd}: $\shift^i r \oddpairs
\shift^j s$ for all $i$ and $j$. Likewise, we define $r \sevenpairs s$
when $S^i r \evenpairs S^j s$ for all $i$ and $j$. As above, when
these two sets do not form a path, $r \sevenpairs s$ is a sufficient
condition for rendezvous.

Thus our strategy shall be to define a map
$x \mapsto \sencode{x}$ with the property that for two strings $x, y$,
\begin{equation}\label{eq:sencode-criterion}
x = y \Rightarrow \sencode{x} \sevenpairs
\sencode{y} \qquad\text{and}\qquad x \neq y \Rightarrow \sencode{x} \soddpairs
\sencode{y}\,.
\end{equation}
With   such  a  map   defined,  the   construction  follows   that  of
the previous construction: the cyclic  schedule adopted by the pair
$(\alpha,\beta)$ is  given by $\sencode{\chi(\alpha,\beta)_2}$  where $\chi$ is  an edge
coloring of $L_n$.

Anticipating  the construction, we  set down  some terminology.  For a
string $z$, we define  the ``graph'' of $z$ to be the function
$G_z: \{ 0, \ldots, |z|\} \rightarrow \Z$ given by
\[
G_z(0)=0, \quad G_z(k) = \sum_{i = 1}^k (2z_i - 1)
\]
so that $G_z$ traces out the  ``walk'' prescribed by $z$ in which each
$1$ corresponds to a step northeast and each $0$ corresponds to a step
southeast as in  Figure~\ref{fig:graph}.  We say that a binary string $z$
is \emph{balanced} if $\wt(z)  = |z|/2$ (so that $|z|$ is necessarily
even);  equivalently $G_z(|z|)  = 0$,  see  Figure~\ref{fig:balanced}. A
balanced string $z$  is \emph{Catalan} if $G_z$ is  never negative. If
$G_z$ is positive, which is to say that  $G_z(i) > 0$ for all $0 < i <
|z|$,   we    say   that    $z$   is   \emph{strictly    Catalan};   see
Figure~\ref{fig:Catalan}. We  remark that  if $z$ is  Catalan, $1\circ
z\circ  0$  is   strictly  Catalan.  Finally,  we  say   that  $z$  is
\emph{$t$-maximal} if  the set $\{ i  \mid G_z(i) =  \max_j G_z(j) \}$
has  size  exactly  $t$;  the  notion  \emph{$t$-minimal}  is  defined
analogously. Note that a  strictly Catalan sequence $z$ is $1$-minimal
and this  single minimum appears  at $i =  0$.  We remark that  if the
string $z$  is $t$-maximal (or $t$-minimal),  the same can  be said of
all shifts of $z$.

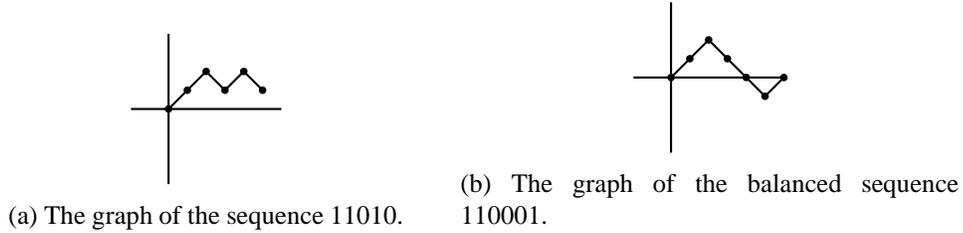
\begin{figure}[ht]
  \centering
  \begin{subfigure}[b]{.4\textwidth}
    \centering
    \begin{tikzpicture}[scale=.25]
      \draw[thick] (0,-4) -- (0,4);
      \draw[thick] (-2,0) -- (6,0);
      \fill[thick] (0,0)  circle (.2);
      \fill[thick] (1,1)  circle (.2);
      \fill[thick] (2,2)  circle (.2);
      \fill[thick] (3,1) circle (.2);
      \fill[thick] (4,2)  circle (.2);
      \fill[thick] (5,1)  circle (.2);
      \draw[thick] (0,0) -- (1,1) -- (2,2) -- (3,1) -- (4,2) -- (5,1);
    \end{tikzpicture}
    \caption{The graph of the sequence $11010$.}
    \label{fig:graph}
  \end{subfigure}
  \begin{subfigure}[b]{.4\textwidth}
    \centering
    \begin{tikzpicture}[scale=.25]
      \draw[thick] (0,-4) -- (0,4);
      \draw[thick] (-2,0) -- (6,0);
      \fill[thick] (0,0)  circle (.2);
      \fill[thick] (1,1)  circle (.2);
      \fill[thick] (2,2)  circle (.2);
      \fill[thick] (3,1) circle (.2);
      \fill[thick] (4,0)  circle (.2);
      \fill[thick] (5,-1)  circle (.2);
      \fill[thick] (6,0)  circle (.2);
      \draw[thick] (0,0) -- (1,1) -- (2,2) -- (3,1) -- (4,0) -- (5,-1) -- (6,0);
    \end{tikzpicture}
    \caption{The graph of the balanced sequence $110001$.}
    \label{fig:balanced}
  \end{subfigure}
  \caption{Graphs and balanced strings.}
\end{figure}

Our strategy is to work with an injective map $\sencode{\cdot}$ with the
property that $\sencode{x}$ is balanced, strictly Catalan, and 2-maximal. Before
describing a construction, we observe that such a map has the
properties outlined in~\eqref{eq:sencode-criterion} above.

Observe, first of all, that if two distinct strings $\sencode{x}$
and $\sencode{y}$ are balanced, it follows immediately that $\sencode{x} \oddpairs \sencode{y}$,
indeed, the number of appearances of $(0,1)$ is the same as the number
of appearances of $(1,0)$ and cannot be zero because the strings are
distinct. Thus, when $\sencode{x}$ and $\sencode{y}$ are balanced, the condition that $\sencode{x}
\not\in \{ \shift^i \sencode{y} \mid i \in [\ell]\}$ is enough to guarantee that
$\sencode{x} \soddpairs \sencode{y}$. Note that if a string $z$ is strictly Catalan, no
nontrivial shift of $z$ can be strictly Catalan. In particular, all
nontrivial shifts of a strictly Catalan string are 1-minimal (as this
is a property enjoyed by strictly Catalan strings) with a different
unique point of minimality. It follows that $x \neq y \Rightarrow \sencode{x}
\soddpairs \sencode{y}$, as desired.


To ensure that $\sencode{x} \evenpairs \sencode{y}$,
when $\sencode{x}$ and $\sencode{y}$ are balanced it suffices to exclude the possibility
that $\sencode{x} = \overline{\sencode{y}}$; similarly, the number of appearances of
$(0,0)$ is the same as the number of appearances of $(1,1)$, and
cannot be zero unless the strings are complements. We conclude that,
for two balanced strings $\sencode{x}$ and $\sencode{y}$, the condition $\sencode{x} \not \in \{
\shift^i \overline{\sencode{y}} \mid i \in [n]\}$ implies that $\sencode{x} \sevenpairs
\sencode{y}$.
Observe that as string $z$ is $k$-maximal if and only if $\overline{z}$ is
$k$-minimal. Thus if $\sencode{x}$ and $\sencode{y}$ are 1-minimal (as they must be if
they strictly Catalan), and 2-maximal, then $\sencode{x} \neq
\overline{\sencode{y}}$. Thus $\sencode{x} \sevenpairs \sencode{x}$ for all $x$, as desired.

It
remains to show that we can efficiently construct such a function.

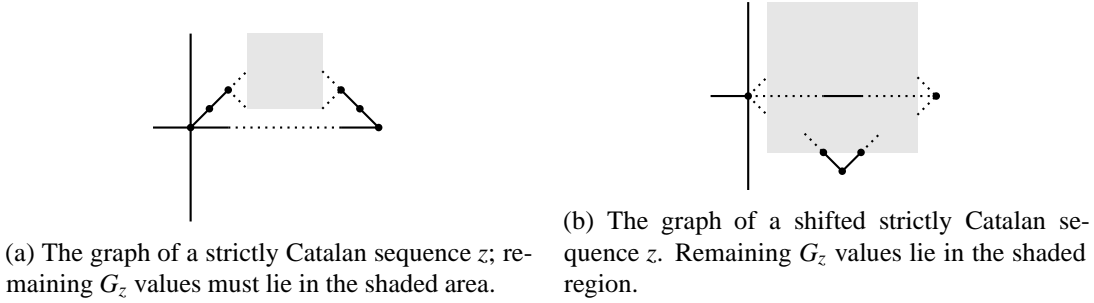
\begin{figure}[ht]
  \centering
  \begin{subfigure}[b]{.42\textwidth}
    \centering
    \begin{tikzpicture}[scale=.25]
      \fill[gray!20] (3,1) rectangle (7,5);
      \draw[thick] (0,-5) -- (0,5);
      \draw[thick] (-2,0) -- (2,0);
      \draw[thick,dotted] (2,0) -- (8,0);
      \draw[thick] (8,0) -- (10,0);
      \fill[thick] (0,0)  circle (.2);
      \fill[thick] (1,1)  circle (.2);
      \fill[thick] (2,2)  circle (.2);
      \fill[thick] (10,0) circle (.2);
      \fill[thick] (9,1)  circle (.2);
      \fill[thick] (8,2)  circle (.2);
      \draw[thick] (0,0) -- (1,1) -- (2,2);
      \draw[thick,dotted] (2,2) -- (3,3);
      \draw[thick,dotted] (2,2) -- (3,1);
      \draw[thick,dotted] (7,3) -- (8,2);
      \draw[thick,dotted] (7,1) -- (8,2);
      \draw[thick] (10,0) -- (9,1) -- (8,2);
    \end{tikzpicture}
    \caption{The graph  of a strictly Catalan  sequence $z$; remaining
      $G_z$ values must lie in the shaded area.}
  \end{subfigure}
\quad
  \begin{subfigure}[b]{.42\textwidth}
    \centering
    \begin{tikzpicture}[scale=.25]
      \fill[gray!20] (1,-3) rectangle (9,5);
      \draw[thick] (0,-5) -- (0,5);
      \draw[thick] (-2,0) -- (0,0);
      \draw[thick,dotted] (0,0) -- (4,0);
      \draw[thick] (4,0) -- (6,0);
      \draw[thick,dotted] (6,0) -- (10,0);
      \fill[thick] (0,0)  circle (.2);
      \fill[thick] (10,0)  circle (.2);
      \fill[thick] (4,-3)  circle (.2);
      \fill[thick] (5,-4)  circle (.2);
      \fill[thick] (6,-3)  circle (.2);
      \draw[thick] (4,-3) -- (5,-4) -- (6,-3);
      \draw[thick,dotted] (3,-2) -- (4,-3);
      \draw[thick,dotted] (6,-3) -- (7,-2);
      \draw[thick,dotted] (0,0) -- (1,-1);
      \draw[thick,dotted] (0,0) -- (1,1);
      \draw[thick,dotted] (9,1) -- (10,0);
      \draw[thick,dotted] (9,-1) -- (10,0);
    \end{tikzpicture}
    \caption{The graph of a shifted strictly Catalan sequence $z$. 
      Remaining $G_z$ values lie in the shaded region.}
  \end{subfigure}
  \caption{Catalan sequences.}\label{fig:Catalan}
\end{figure}

  \begin{figure}[ht]
  \centering
  \begin{subfigure}[b]{.42\textwidth}
    \centering
    \begin{tikzpicture}[scale=.25]
      \begin{scope}
        \fill[gray!20] (1,-5) rectangle (9,4);
        \draw[thick] (0,-5) -- (0,5);
        \draw[thick] (-2,0) -- (0,0);
        \draw[thick,dotted] (0,0) -- (4,0);
        \draw[thick] (4,0) -- (6,0);
        \draw[thick,dotted] (6,0) -- (10,0);
        \fill[thick] (0,0)  circle (.2);
        \fill[thick] (10,0)  circle (.2);
        \fill[thick] (4,3)  circle (.2);
        \fill[thick] (5,4)  circle (.2);
        \fill[thick] (6,3)  circle (.2);
        \draw[thick] (4,3) -- (5,4) -- (6,3);
        \draw[thick,dotted] (3,2) -- (4,3);
        \draw[thick,dotted] (6,3) -- (7,2);
        \draw[thick,dotted] (0,0) -- (1,-1);
        \draw[thick,dotted] (0,0) -- (1,1);
        \draw[thick,dotted] (9,1) -- (10,0);
        \draw[thick,dotted] (9,-1) -- (10,0);
      \end{scope}
    \end{tikzpicture}
    \caption{The graph of a sequence, showing a maximum value.}
  \end{subfigure}
  \begin{subfigure}[b]{.42\textwidth}
    \centering
    \begin{tikzpicture}[scale=.25]
    \begin{scope}[shift={(0,-12)}]
      \draw[thick] (0,-5) -- (0,5);
      \draw[thick] (-2,0) -- (0,0);
      \draw[thick,dotted] (0,0) -- (4,0);
      \draw[thick] (4,0) -- (10,0);
        \draw[thick,dotted] (10,0) -- (14,0);
        \fill[thick] (0,0)  circle (.2);
        \fill[thick] (14,0)  circle (.2);
        \fill[thick] (4,3)  circle (.2);
        \fill[thick] (5,4)  circle (.2);
        \fill[thick,blue] (7,4)  circle (.2);
        \fill[thick,blue] (8,5)  circle (.2);
        \fill[thick,blue] (9,4)  circle (.2);
        \fill[thick,blue] (6,5)  circle (.2);
        \fill[thick] (10,3)  circle (.2);
        \draw[thick] (4,3) -- (5,4);
        \draw[thick] (9,4) -- (10,3);
        \draw[thick,blue] (5,4) -- (6,5) -- (7,4) -- (8,5) -- (9,4);
        \draw[thick,dotted] (3,2) -- (4,3);
        \draw[thick,dotted] (10,3) -- (11,2);
        \draw[thick,dotted] (0,0) -- (1,-1);
        \draw[thick,dotted] (0,0) -- (1,1);
        \draw[thick,dotted] (13,1) -- (14,0);
        \draw[thick,dotted] (13,-1) -- (14,0);
      \end{scope}
    \end{tikzpicture}
    \caption{The sequence after the transformation to $2$-maximality.}
  \end{subfigure}
  \caption{The transformation to $2$-maximality.}
  \label{fig:two-maximal}
\end{figure}
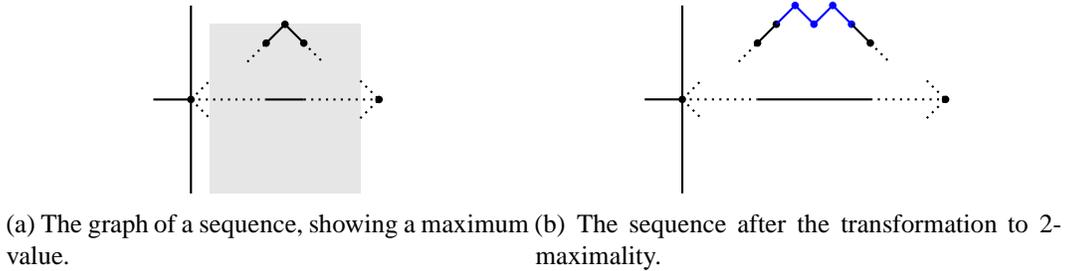
Our  starting  point  shall   be  the  ``Knuth  mapping''  $x  \mapsto
\knuth{x}$ on all the binary strings; this is  an efficient, injective mapping with the
property that $\knuth{x}$ is balanced; moreover,
\[
|\knuth{x}| \leq |x| + \ilog |x| + (1/2) \ilog \ilog |x|\,.
\]
(See~\citet{Knuth:1986fk} for further discussion.) Observe that if $z$
is  balanced,  there is  at  least one  shift  $\shift^c  z$ which  is
Catalan. To yield an invertible process, we consider the map
\[
\unique{z} = (\shift^c z) \circ \underbrace{1^{\ell/2} \circ \knuth{c_2} \circ 0^{\ell/2}}_{(*)}\,,
\]
where $\ell = |\knuth{c_2}|$. Note that the string $(*)$ is Catalan, as $\knuth{c_2}$
is balanced and hence has no more than $\ell/2$ zeros. It follows that
$\unique{z}$ is Catalan (as the concatenation of two Catalan strings
is Catalan). Since the shift $c$ is encoded into $\unique{\cdot}$, the
function is clearly injective. It follows that the map $z \mapsto 1
\circ \unique{\knuth{z}} \circ 0$ is invertible, and carries $z$ to a
\emph{strictly} Catalan image. Finally, we observe that inserting the
string $1010$ at any maximal point in a string $z$ transforms it into
a $2$-maximal string in an invertible fashion (and preserves the other
properties we care about).  We let $\maximal{z}$ denote this
transformation; see Figure~\ref{fig:two-maximal}. To complete the
story, we define
\[
\sencode{z} \triangleq \maximal{1 \circ \unique{\knuth{z}} \circ 0}
\]
and observe that $|\sencode{z}| \leq |z| + 4 \ilog |z| + 16$.
Since $z$ is an edge color with length $\ilog \ilog n$, the theorem is proved.
\end{proof}

\subsection{A general \texorpdfstring{$n$}{n}-schedule}
In this  section we  show how  to apply the  previous result  to yield
$n$-schedules  that provide  rendezvous  in time  $O(|A| |B|  \log\log
n)$. Specifically, we prove the following theorem.

\begin{theorem}\label{thm:upper-general}
  There  is  an  $n$-schedule  so that  for  all
  overlapping  $A,  B \subseteq  [n]$,  the  schedules $\sigma_A$  and
  $\sigma_B$  rendezvous asynchronously in  time  $O(|A| \cdot  |B| \log\log
  n)$.
\end{theorem}

\begin{proof}
  Consider a set $A = \{ a_0, \ldots, a_{k-1}\}$. The schedule for $A$
  depends on a pair of primes $p, p'$ in the range $[k, 3k]$ (there
  always exist two primes in this range). We then construct a schedule
  consisting of a sequence of \emph{epochs}, where the $r$th epoch
  calls for the size-two schedule of Theorem~\ref{thm:upper-2} involving
  the two channels $a_i$ and $a_j$, where $i \equiv r \mod{p}$
  and $j \equiv r \mod{p'}$. (If either $i$ or $j$ do not fall in the range
  $\{0, \ldots, k-1\}$, then we choose an arbitrary element of $A$ to
  fill its place.)

  In the following, we will say a pair of prime numbers $(p, q)$ is \emph{helpful} for the rendezvous of two agents $A$ and $B$ if:
\textsl{(i.)} $p$ is one of the primes selected by the first agent as
described above, \textsl{(ii.)} $q$ is one of the primes selected by the second agent as described
    above, and \textsl{(iii.)} $p\ne q$.
  The construction above specifies that each agent must choose two
  primes to ensure that any two agents are guaranteed to have a \emph{helpful} pair between them.

  Now, suppose $A\cap B  = \{c\}$, and that $c = a_x = b_y$ (so that
  $c$ is the $x^{th}$ channel in
  $A$ and the $y^{th}$ channel in $B$).
  In  the synchronous model,  we  use the  construction described  in
  the proof of Theorem~\ref{thm:upper-2} to  get a schedule  for $(a_i,
  a_j)$ in each epoch. In this case, it suffices to show that there is
  an epoch $r$ satisfying  $r \equiv x \pmod{p}$
  and $r \equiv  y \pmod{q}$, where $p$ and $q$ are a helpful
  pair as described above.  According to  the Chinese Remainder
  Theorem,  there exists  a  solution for  $r$  that is  no more  than
  $pq$. Therefore, in  the worst case, the two agents will both access
  the common  channel at one time, no later than $pq(\log\log    n+   
  \log\log\log n+2)=O(k\ell\log\log n)$ steps after their schedules
  commence.

  The asynchronous model requires only a slight modification.
  Suppose that, for a given epoch, $r$, an agent using the scheme described
  immediately above with subset $A$ executes schedule $\sigma^r_A$
  of length $R$ (all epochs have the same length). 
  Then, we can handle 
  asynchronous rendevous by doubling the length of each epoch and
  executing $\sigma^r_A \sigma^r_A$.
  Assume the commencement time for the $\sigma_A$  is to be $t_a$ and
  the commencement time for $\sigma_B$ is to be $t_b$ where, without
  loss of generality, $t_a \leq t_b$.
  Let $\mu$ denote the closest integer to $\frac{t_b-t_a}{2R} $.
  Then for any $r$, the $r^{th}$ epoch of $\sigma_A$ will overlap with
  the $(r-\mu)^{th}$ epoch of $\sigma_B$ by at least $R$ timesteps.
  For any $r$ such that $r \equiv x \pmod{p}$ and $(r-\mu) \equiv y
  \pmod{q}$,  where the pair $(p,q)$ is helpful,
  then the $r^{th}$ epoch of $\sigma_A$ will overlap with the $(r-\mu)^{th}$
  epoch of $\sigma_B$ no less than $R$ timeslots. Since
  Theorem~\ref{thm:upper-2} guarantees rendezvous between $
  \sigma^r_A$ and any cyclic shift of $\sigma^{r-\mu}_B$, this overlap
  must contain such a rendezvous point.
  
  Again by the  Chinese Remainder Theorem, we know  that  there exists
  a  epoch  $r$ such  that $r-\mu$ is no more than $pq$.  Therefore,
  in the worst case, the two agents will  access the same channel in time
  $2pqR=O(k\ell \log\log n)$ after $t_b$.
\end{proof}

\subsection{A general reduction that guarantees fast symmetric rendezvous}
\label{section:symmetric-reduction}

The rendezvous literature has given special attention to the
\emph{symmetric} case, where $A = B$. For a general
schedule that guarantees rendezvous for all (perhaps distinct) pairs
of sets, one specifically examines the rendezvous time in this
symmetric case. In this section, we observe that any schedule that
guarantees rendezvous for all pairs of sets can be transformed into
one that additionally guarantees $O(1)$ rendezvous time in the
symmetric case, at the expense of a constant blow-up in the rendezvous
time for all other pairs of sets.

Specifically, for a family of schedules $\Sigma = (\sigma_A)_{A
  \subset [n]}$, for each $A \subset [n]$, we define a new schedule
$\hat{\sigma}_A$ as follows: when $\sigma_A$ calls for the channel
$c_1$, $\hat{\sigma}_A$ carries out a short sequence of accesses,
consisting of the channel $c_1$ and the channel $c_0 = \min \{A\}$
(the smallest element of $A$) in the pattern $c_0 c_1 c_0 c_0 c_1 c_1$
repeated twice. The significance of this pattern is that $010011
\sevenpairs 010011$: thus any pair of rotations of $c_0 c_1 c_0 c_0
c_1 c_1$, will yield simultaneous accesses to both $(c_0, c_0)$ and
$(c_1,c_1)$. To ensure that there is sufficient overlap in these short
sequences of accesses, we repeat them twice: as in the proof of
Theorem~\ref{thm:upper-general}, this guarantees that a full rotation
of the sequence overlaps. By a similar argument, it follows that the
time to rendezvous, for any pair of sets, is no more than a constant
factor (12, by this construction) larger than in $\Sigma$. However,
when $A = B$, such a pair will rendezvous (at their smallest element)
in constant time.

\section{Lower bounds}\label{sec:lower-bounds}

In this section we establish that
\begin{enumerate}
\item  $R_s(n,k) = \Omega(\log\log n)$ for any $k \leq n/2$. (Theorem~\ref{thm:lower-2} and Corollary~\ref{cor:lower-wk-k}.)
\item\label{item:large-k-common} $R_s(n,k) \geq k^2$ for all $k = O(\log n/\log\log n)$ and,
  in general, $R_s(n,k) \geq \alpha k$ for all $k \leq
  {n}^{{1}/{2\alpha}}$ (so long as $\alpha \leq k$). (Theorem~\ref{thm:lower-k}.)
\item\label{item:large-k-relative} $R_a(n,k) \geq k^2$ for all $2
  \leq k \leq n/2$. (Theorem~\ref{thm:lower-relative-k}.)
\end{enumerate}
The lower bounds provided by items~\ref{item:large-k-common}
and~\ref{item:large-k-relative} exhibit an enormous gap for large $k$
and, indeed, the behavior of
$R_s(n,k)$ and $R_a(n,k)$ must diverge for $k \approx \sqrt{n}$. In particular, $R_a(n,k) =
\Omega(k^2)$ while there is a simple algorithm that shows that
$R_s(n,k) \leq n$ for all $k$: each
agent hops on channel $t$ at time $t$ when $t$ is in the channel
set, and remains silent otherwise.

\paragraph{The dependence of rendezvous time on $n$.} We begin with two lower bounds that establish that
$R_s(n,k) \rightarrow \infty$ as $n \rightarrow \infty$.

\begin{theorem}\label{thm:lower-2}
  For all $n \geq 2$, $R_s(n,2) = \Omega(\log\log n)$.
  Rendezvous requires at least $\Omega(\log \log n)$ time, even
  in the synchronous model when agents are promised to have sets of size 2.
\end{theorem}

\begin{proof}
Consider the complete graph $K_n$, with the interpretation that each
vertex represents a channel and each edge represents a set of size
two. In this case where agents correspond to two channels, we
represent schedules as binary sequences, $s \in \{0,1\}^{\N}$, with the
convention that a $0$ calls for hopping on the smaller channel and $1$
calls for hopping on the larger channel.

Let $\Sigma$ be an
$(n,2)$-schedule which guarantees rendezvous synchronously in $T$. In
this case, we may treat each $\sigma_{(i,j)}$ as a finite length
string in $\{0,1\}^T$, with the understanding that rendezvous is
guaranteed before any schedule is exhausted. Treat the schedules $\sigma_{(i,j)} \in \{0,1\}^T$ as a coloring of the
edges of $K_n$. According to a variant of Ramsey's theorem, any
$m$-coloring of the edges of the complete graph must have a
monochromatic triangle when $n
\geq e m!$. (See, e.g.,~\cite{Graham1990Ramsey}.) Note, however, that
a monochromatic triangle yields, in
particular, an ordered triple $i < j < k$ for which the schedules
associated with $(i,j)$ and $(j,k)$ are identical; such schedules
never rendezvous. It
follows that $e (2^T)! \geq n$ and, by Sterling's estimate $x! \sim
\sqrt{2 \pi x}(x/e)^x$ that $T = \Omega(\log\log n)$.
\end{proof}

\begin{corollary}\label{cor:lower-wk-k}
  For any $k  \leq n/2$, $R_s(n,k) = \Omega(\log \log n)$.
\end{corollary}

\begin{proof}
  Write $[n]$ as the disjoint union of two sets $A = \{1, \ldots, m\}$
  and $B = \{ m+1, \ldots, n\}$, where $|B| \geq |A| (k-2) = m(k-2)$;
  our strategy will be to extend the sets of size two in $A$ to a
  family of subsets of $[n]$ of size $k$ in such a way that schedules
  for these extended sets can be ``pulled back'' to schedules for the
  sets of size two (for which the previous lower bound applies). To
  proceed with this idea, we express $B$ as a disjoint union
  $B = \left( B_1 \cup \cdots \cup B_m \right) \cup B_{\text{rest}}$,
  where each $B_i$ has size exactly $k-2$. Now, we consider the $\binom{|A|}{2}$ 
  sets of the form
  \[
  X_{\{i,j\}} \triangleq \{ i, j\} \cup B_{i + j \bmod m}\,,
  \]
  where $i, j \in A$. Let $\Sigma$ be an $(n,k)$-schedule. Observe that a schedule $\sigma_{X_{\{i,j\}}}$ for the set
  $X_{\{i,j\}}$ can be treated as schedule $\check{\sigma}_{\{ i,
    j\}}$ (for $\{i,j\}$) by \emph{restriction}, simply replacing all
  references to elements outside $\{ i, j\}$ with, say, the smaller of
  $i$ and $j$. In general, restriction of an $(n,k)$-schedule to an
  $(n,\ell)$-schedule (for $\ell < k$) does not provide any guarantee
  on rendezvous, even when the original $(n,k)$-schedule
  does. However, the intersection pattern of the sets $X_{\{i,j\}}$
  above is chosen in such a way that the $(m,2)$-schedule
  $\check{\Sigma}$ obtained by defining $\check{\sigma}_{i,j}$ to be
  the restriction of the schedule $\sigma_{X_{\{i,j\}}}$ will
  guarantee rendezvous.

  Consider two subsets $\{ i, j\}$ and $\{ {i'}, {j'}\}$ of $A$,
  each of size two. If these two sets are not identical but share a
  common element, it follows that $i + j \bmod m \neq i' + j' \bmod m$.
  Thus,
  \begin{align}
  &B_{i + j \bmod m} \cap B_{i' + j' \bmod m} = \emptyset \quad \nonumber \\
  &\ \ \ \ \ \ \ \ \ \ \ \text{and} \quad \nonumber \\
  &X_{\{i, j\}} \cap X_{\{{i'}, {j'}\}} = \{ i, j\} \cap \{{i'}, {j'}\} \nonumber
  \end{align}
  If $\sigma_{X_{i,j}}$ and $\sigma_{X_{i',j'}}$ rendezvous, this must
  occur at a channel in $\{i,j\} \cap \{i', j'\}$, and it follows that
  the rendezvous time of the schedule $\Sigma$ is at least that
  of the schedule $\check\Sigma$; we conclude that $R(n,k) \geq R(m,2)$ so
  long as $n \geq m + m(k-2) = m(k-1)$. Thus $R(n,k) \geq
  R(\lfloor(n/(k-1)\rfloor, 2) = \Omega(\log\log n/k)$.

  However, it is clear that $R_s(n,k) \geq k$ for all $k \leq n/2$, so
  the bound above is only relevant when $k = \Omega(\log\log n)$ which
  yields a $\Omega(\log \log n)$ lower bound for all $k$.
\end{proof}

\paragraph{The dependence of rendezvous time on $k$ in the synchronous setting.}
\begin{theorem}\label{thm:lower-k}
  Let $1 \leq \alpha \leq k$ and $k \leq n^{1/(2\alpha)}$. Then
  $R_s(n,k) \geq k\alpha$. In particular, for $k = O(\log n/\log\log
  n)$, $R_s(n,k) \geq k^2$.
\end{theorem}

\begin{proof}
  Let $\Sigma$ be an $(n,k)$-schedule. Partition the $n$ channels into
  $n/k$ disjoint subsets, $S_1, \ldots, S_{{n}/{k}}$, each of size
  $k$. Suppose, for the sake of contradiction, that $\Sigma$ guarantees rendezvous
  synchronously in less than $\alpha k$. In this case, we focus only on the
  first $\alpha k - 1$ time slots of the schedules and treat each
  $\sigma_A$ as a function defined on $\{1, \ldots, \alpha k - 1\}$. 

  For each $i \in \{ 1,\ldots,{n}/{k} \}$, let $\sigma_i$ denote the
  schedule of subset $S_i$ and observe that some $a_i \in
  S_i$ must appear fewer than $\alpha \leq k$ times in the
  schedule. Letting $\sigma_i^{-1}(a_i) \subseteq \{1,2,\cdots,
  \alpha k -1\}$ denote the set of time indices at which $a_i$ appears in
  $\sigma_i$, we then have $|\sigma_{i}^{-1}(a_{i})| < \alpha$. By
  possibly adding some elements to the set
  $\sigma_{i}^{-1}(a_{i})$, we may construct a set $A_i$,
  containing $\sigma_{i}^{-1}(a_{i})$, of size exactly $\alpha
  -1$.  Observe that there are $\binom{\alpha k - 1}{\alpha -1}$ possible
  values (subsets) that these $A_i$ can assume.
 
  If ${n}/{k}$, the number of disjoint subsets in our original
  partition, exceeds $(k-1)\cdot \binom{\alpha k-1}{\alpha-1}$, then
  there must be at least $k$ of these subsets, say $S_{i_1}, \ldots,
  S_{i_k} $, for which 
  \[
  A_{i_1}
  =\cdots= A_{i_k} = Z\,,
  \]
  for a set $Z$ of size $\alpha - 1 < k$; it follows that
  $\sigma_{i_j}^{-1}(a_{i_j}) \subset Z$ for each $i$.

Finally, let $\hat{S} =\{a_{i_1},
  \ldots, a_ {i_k}\}$ and let $\hat \sigma$ be its schedule in
  $\Sigma$.  For any $j\in \lbrace 1, \ldots, k\rbrace$, $\hat \sigma$
  must rendezvous with $\sigma_{i_j}$, which requires that ${\hat
    \sigma}^{-1}(a_{i_j}) \cap Z \neq \emptyset$. As the
  $\hat{\sigma}^{-1}(a_{i_j})$ are disjoint, this implies that $|Z|
  \geq k$, a contradiction. To satisfy the condition that ${n}/{k} >
  (k-1)\binom{(\alpha-1) k}{\alpha-1}$, it suffices for
  \[
  n \geq k^{2\alpha}\,,
  \]
  where we have applied the coarse bound $\binom{\alpha k-1}{(\alpha-1)}
  \leq \binom{k^2}{\alpha-1}
  \leq k^{2(\alpha-1)}$.
\end{proof}

\paragraph{A stronger lower bound in the asynchronous model.}

Finally, we show that in the asynchronous model, it is possible to
extend the $k^2$ lower bound to all $k$ less than $n/2$. In fact, we
show that in any $n$-schedule, for any $k$ and $\ell$ with $k + \ell \leq
n + 1$ there are sets of size $k$ and $\ell$
that cannot rendezvous asynchronously in time less than $k\ell$.
\begin{theorem}\label{thm:lower-relative-k}
  For all $k \leq n/2$, $R_a(n,k) \geq k^2$. Moreover, for any
  $n$-schedule and any $k$ and $\ell$ for which $k + \ell \leq n+1$,
  there are sets of size $k$ and $\ell$ that require at least $k\ell$
  steps to rendezvous in the asynchronous model.
\end{theorem}

\begin{proof}
  Let $\Sigma$ be  an $n$-schedule. We will show  that there exist two
  subsets, $A$  and $B$, such  that $|A| =  k$, $|B|= \ell$,  $|A \cap
  B|=1$, and  $\sigma_A$ and $\sigma_B$ require at  least $k\ell$ time
  steps to rendezvous in the asynchronous model.
  First, consider uniformly random selection of $A, B \subset [n]$
  according to the following process: \textsl{(i.)} select $A$
  uniformly among all the sets of size $k$,  \textsl{(ii.)} select a
  channel $h$ uniformly from $A$, and \textsl{(iii.)} select $B'$
  uniformly at random from all subsets of $[n]\setminus A$ of size
  $\ell-1$ and define $B = B' \cup \{h\}$. We remark that the reversing
  roles of $A$ and $B$ in the above process (initially
  selecting $B$ uniformly among all sets of size $\ell$, selecting $h$
  from $B$, and selecting $A$ by adding $k-1$ random elements of $[n]
  \setminus A$ to $\{h\}$) yields the same
  probability distribution on $(A, B)$.
  
  We let $\Delta(h, \sigma; T)$ denote the density of occurrences
  of $h$ during the first $T$ time steps in schedule $\sigma$:
  $\Delta(h, \sigma; T) \triangleq {|\{t\in [0, T) \mid
    \sigma(t)=h\}|}/{T}$.  (Here the notation $[0,T)$
  denotes $\{ 0, 1, \ldots, T-1\}$.)  For any length-$T$ prefix of the
  schedule $\sigma_A$ for $A$, note that
  \begin{align}
  \Exp\limits_{\substack{A, h\in A}}[\Delta(h, \sigma_A; T)] 
  &= \Exp\limits_{A}\left[\sum \limits_{x\in A} \Pr(h=x)\Delta(x, \sigma_A; T)\right] \nonumber \\ 
  &=\Exp\limits_{A}\left[\frac{1}{k}\sum \limits_{x\in A}\Delta(x, \sigma_A; T)\right] \nonumber \\
  &= \frac{1}{k}\,.
 \nonumber
  \end{align}
  (Here $\Exp[\cdot]$ denotes expectation). Likewise, considering the
  reversed procedure for selecting $A$ and $B$, for any $T'$ we have
  $\Exp[\Delta(h,\sigma_B; T')] = {1}/{\ell}$.
  By linearity of expectation, for any $T, T'$,
  \begin{equation} \label{eq:linexp}
    \Exp\limits_{A, B, h}[k\cdot\Delta(h, \sigma_A; T)+ \ell\cdot\Delta(h, \sigma_B;
    T')]= 2\,.
  \end{equation}
  
  Let $r$ be the minimum integer so that all intersecting subsets, $A$
  and $B$ of sizes $|A| = k$ and $|B| = \ell$, intersect in 
  time $r$; let $R \gg r$. From the expectation
  calculation~\eqref{eq:linexp} it follows that there exist two sets,
  $A$ and $B$, intersecting at an unique element $h$, for which
  $k\Delta(h, \sigma_A; R)+ \ell\Delta(h, \sigma_B; r) \leq
  2$. Observe then that the product
  \[
  k \Delta(h, \sigma_A; R) \cdot
  \ell \Delta(h, \sigma_B; r) \leq 1
  \]
  and hence $\Delta(h, \sigma_A; R) \cdot
  \Delta(h, \sigma_B; r) \leq 1/k\ell$.

  Consider, finally, the circumstances when the schedule $\sigma_A$
  starts at time $0$ and the schedule $\sigma_B$ starts at some time
  $t \in [0, R-r]$.  Let $P=\{(x, y)\in[0, R)\times [0, r) \mid
  \sigma_A(x)=\sigma_B(y)=h, x \geq y\}$. Each such pair $(x,y)$ is a
  possible rendezvous point which can occur only if $\sigma_B$ starts
  at time $x-y$. We have
  \[
  |P|\leq R\cdot
  \Delta(h,\sigma_A; R) \cdot r\cdot \Delta(h,
  \sigma_B; r)\leq \frac{R\cdot r}{k\ell}\,.
  \]
  As rendezvous is guaranteed in the range $[t, t+r)$ for any $t \in [0,
  R-r]$, we  must have $|P| \geq  R-r$ (otherwise, there is a
  time  that is  not covered  by any  rendezvous pair  of  $P$), which
  implies  that ${R\cdot  r}/{k\ell}
  \geq     R-r$    and,     therefore,
  \[
  r     \geq     \frac{R-r}{R} \cdot k\ell\,.
  \]
  As $R\rightarrow \infty$, this quantity approaches $k\ell$.
\end{proof}

\section{Rendezvous with a one-bit beacon}
\label{sec:one-bit}

In this section we consider the rendezvous problem when the agents are
supplied with a ``one-bit random beacon.'' Specifically, we work under
the assumption that  the agents exist in an  environment that supplies
them with  a (common)  uniformly random bit  $c_t \in  \{0,1\}$ during
each  time step $t$;  we assume  that the  $c_t$ are  independent (for
different $t$) and  available to all agents. We  remark that random beacons
have  been  studied  in  a  number  of  related  models~\cite{Rabin83,
  DziembowskiM02} and--in  practice--beacons are available,  e.g., for
GPS receivers in close proximity~\cite{GPS, LeeCC05}.

We shall see that augmenting the basic model with a one-bit beacon can
dramatically reduce the rendezvous time: in particular, with a one-bit
beacon, (asynchronous) rendezvous is possible with high probability in
time  $O(|S_i|  +  |S_j|   +  \log  n)$.  (In  contrast,  asynchronous
rendezvous,  without  such   a  beacon,  requires  time  $\Omega(|S_i|
|S_j|)$.)

For  a number  $n$,  we let  $\mathfrak{S}_n$  denote the  set of  all
permutations  of   the  elements  $\{1,  \ldots,  n\}$,   the  set  of
channels. The schedule for an  agent $i$ with available channels $S_i$
is constructed as follows:
\begin{itemize}
\item  At  time  $t$, the  sequence  $c_1,  \ldots,  c_t$ is  used  to
  determine a permutation $\pi_t \in \mathfrak{S}_n$. (We write $\pi_t
  = \Pi(c_1  \ldots c_t)$, and  discuss below various choices  for the
  function $\Pi$.)
\item The agent hops on  the channel $\arg \min_{a \in S_i} \pi_t(a)$,
  which is to say that the agent  hops on the channel that maps to the
  smallest element of $\{1,\ldots, n\}$ under the permutation $\pi_t$.
\end{itemize}
It remains to describe $\Pi$, the rule that determines the permutation
$\pi_t$ from  the sequence $c_1,  \ldots, c_t$.  For this  purpose, we
recall the notion of a \emph{min-wise family of permutations}.
\begin{definition}
  We   say    that   a   subset   $R    \subset   \mathfrak{S}_n$   is
  \emph{$\epsilon$-min-wise  independent}  if,  for  every  subset  $A
  \subset \{1, \ldots, n\}$ and every element $a \in A$,
  \[
  \Pr_{\pi \in R}[ \pi(a) = \min \{ \pi(a') \mid a' \in A\}]
  \geq \frac{1}{|A|} (1 - \epsilon)\,.
  \]
  (Here $\pi$ is given the uniform distribution in $R$.)
\end{definition}

For any $n$ and  $\epsilon$, Indyk~\cite{Indyk:2001} gave an efficient
construction   of   a   family   of   $\epsilon$-minwise   independent
permutations  that  can  be  represented  with $O(\log  n  \cdot  \log
1/\epsilon)$  bits. In  our setting,  it suffices  to set  $\epsilon =
1/2$; for the remainder of this  section, we let $R_n$ denote a family
of  $1/2$-minwise independent  permutations in  $\mathfrak{S}_n$. Note
that $d  \log n$ bits are  required to represent an  element in $R_n$,
for a fixed constant $d$.

Consider  now two  sets of  channels $S_i$  and $S_j$  and  an element
$\alpha \in S_i  \cap S_j$. If $\pi$ is a  permutation drawn at random
from $R_n$, then
\begin{align}
  &\Pr\bigl[\alpha = \arg \min_{a \in S_i} \pi(a) = \arg \min_{a' \in
    S_j} \pi(a')\bigr]\nonumber\\
  =&\Pr\bigl[\alpha = \arg \min_{a \in S_i \cup S_j} \pi(a) \bigr]
  \geq \frac{1}{2(|S_i| + |S_j|)}\,. \label{eq:one-perm}
\end{align}

\smallskip
\noindent
\textsl{A  simple  $O(\log  n   \cdot  (|S_i|  +  |S_j|))$  rendezvous
  protocol.}  Let  us  consider   the  protocol  induced  by  defining
$\Pi(c_1\ldots c_t)$  to be the  permutation from $R_n$  determined by
the last $d  \log n$ bits of  $c_1\ldots c_t$. At times 
\[
d  \log n, 2d \log  n,  \ldots,  T  d\log   n\,,
\]
 these  selections  from  $R_n$  are
independent.  In light  of~\eqref{eq:one-perm},  the probability  that
each of these permutations failed to induce rendezvous is no more than
\[
\left(1 - \frac{1}{2(|S_i| + |S_j|)}\right)^{T} \leq
e^{-T/(2(|S_i|+|S_j|))}\,.
\]
It follows that for $T = 2\alpha \ln n \cdot (|S_i| + |S_j|))$, the
probability that this protocol fails to rendezvous is no more than
$e^{-\alpha \ln n} = 1/n^{\alpha}$, as desired.

\smallskip
\noindent
\textsl{An  $O(|S_i| +  |S_j|  + \log  n)$  rendezvous protocol.}  The
protocol  above  can   be  improved  by  applying  \emph{deterministic
  amplification}.  The  protocol  described  above  uses  $O(\log  n)$
independent  random   bits,  essentially,  to  produce   a  family  of
independent elements  of $R_n$. By  ``walking on an  expander graph,''
one can achieve  the same performance guarantees with  only $O(|S_i| +
|S_j|  +  \log  n)$  random  bits. Specifically,  one  associates  the
elements  of the  set $R_n$  with  the vertices  of a  constant-degree
expander graph and generates a  collection of elements of $R_n$ by the
following process:  the first  $d \log  n$ bits of  $c_i$ are  used to
generate  a random  element  of  the expander  graph  (and, hence,  an
element of  $R_n$); each subsequent  element of $R_n$ is  generated by
using $O(1)$ bits of the string $c_1, c_2, \ldots$ to take one step in
the  natural random  walk on  the graph.  See~\cite{Hoory:2006}  for a
survey of these  techniques and, in particular, a  description of this
particular form of deterministic amplification.

\bibliographystyle{plainnat} 

\appendix[\texorpdfstring{$0.439$}{0.439}-approximation for one-round graphical rendezvous]

In the rest  of the paper we considered the  problem of minimizing the
number  of rounds  need to  achieve  rendezvous. In  this appendix  we
consider the  problem of maximizing the  number of pairs  of agents that
achieve  rendezvous in  a single  round in  the graphical  case, i.e.,
where all channel sets are of size two.

In the graphical case each agent  can be viewed as an edge between the
corresponding channels  (vertices). The decision of an  agent $(i, j)$
to select channel $i$ can be  viewed as orienting the edge from $j$ to
$i$. And  a pair  of agents achieves  rendezvous if  the corresponding
arcs (oriented  edges) point inwards  towards the same vertex.   It is
easy to see  that the one-round problem can  be viewed equivalently as
the problem of orienting each edge  in a given graph so as to maximize
the number of pairs of directed edges pointed towards the same vertex.
Let us call  such a pair an  in-pair. Similarly let us call  a pair of
edges oriented outwards from the same vertex an out-pair. And pairs of
incident edges that are oriented differently are termed cross-pairs.

Consider the  scheme where each  edge is oriented uniformly  at random
(in one  of the two possible  orientations). A pair  of edges incident
(at  a  vertex)  will  both  point  towards  the  shared  vertex  with
probability  $\frac{1}{4}$.   Thus,   this  simple  randomized  scheme
achieves rendezvous  between $\frac{1}{4}$  of all possible  pairs and
hence is a $0.25$-approximation algorithm.

We now  present a $0.439$-approximation algorithm based  on rounding a
semi-definite  program  (SDP). Our  semi-definite  program is  closely
related  to the  famous  Goemans-Williamson (GW)  program for  MAX-CUT
\cite{GoemansW95}.   Initially,  we  orient  each edge  of  the  graph
arbitrarily. We associate a vector $\vec{e}$ with each (oriented) edge
$e$  of  the  graph  in  the  SDP\@.  One  can  think  of  $\vec{e}$  as
representing  the initial  orientation  of $e$  and  of $-\vec{e}$  as
representing  the opposite  orientation. Note  that this  is different
from the  GW SDP  which associates  a vector with  each vertex  of the
graph.  We say that a pair of vectors is incident if the corresponding
edges are incident.  Now, for  each pair of incident vectors $\vec{e},
\vec{f}$ we associate a sign $\mbox{sgn}_{(\vec{e},\vec{f})}$ which is
$+1$ if the two vectors form an in-pair or out-pair and $-1$ otherwise
(i.e., a cross-pair).

Consider maximizing the following SDP:

\[ \sum_{|e \cap f| = 1} \frac{1 + \mbox{sgn}_{(\vec{e},\vec{f})}*\vec{e}\cdot\vec{f}}{2} \]

Observe that  if the above SDP  were solved over $(-1,  +1)$ then each
term  contributes $1$ if  the corresponding  vectors are  oriented the
same way with  respect to the incident vertex  and $0$ otherwise. Thus
the  above SDP,  if solved  over $(-1,  +1)$ maximizes  the  number of
in-pairs plus the number of out-pairs.

We    solve    the     above    SDP    using    standard    techniques
\cite{GoemansW95}. The vectors in the resulting solution will lie on a
sphere. We  round by choosing  a random hyperplane and  preserving the
orientation of  edges that fall  in one hemisphere while  flipping the
orientation of edges that fall  in the other hemisphere. The above SDP
is basically the GW SDP (with vectors representing edges as opposed to
vertices) and hence an analysis identical to that in \cite{GoemansW95}
yields a $0.878$-approximation to  the problem of maximizing (over all
orientations) the number of in-pairs plus the number of out-pairs. But
this is at least as much as the maximum number of in-pairs achievable.
However, to  achieve both  in-pairs and out-pairs  it is  necessary to
have  two rounds, the first  with the  normal orientation  (i.e., each
agent  selects  the channel  that  its  corresponding  arc is  pointed
towards), and the second with flipped orientations of all edges. Hence
one of the  two rounds must achieve $\frac{1}{2}\times  0.878 = 0.439$
of   the   maximum   number    of   in-pairs   achievable   over   all
orientations. This scheme can be derandomized yielding a deterministic
$0.439$-approximation to the problem of maximizing the number of pairs
of agents achieving rendezvous in one round in the graphical case.
\end{document}